\documentclass{endm}
\usepackage{endmmacro}
\usepackage{graphicx}
\usepackage[utf8]{inputenc}
\usepackage[T1]{fontenc}
\usepackage{url}
\usepackage{ifthen}
\usepackage[cmex10]{amsmath} 
\usepackage{amssymb, amsbsy}

\newcommand{\NN}{\mathbb{N}}
\newcommand{\F}{\mathbb{F}}

\newcommand{\OMul}[1]{\mathcal{M}(#1)}
\newcommand{\BigO}[1]{O\left(#1\right)}
\newcommand{\BigOtext}[1]{O(#1)}
\newcommand{\M}{\mathbf{M}}
\newcommand{\Red}{\mathrm{R}}
\newcommand{\R}{\F[x,\sigma]}
\newcommand{\mdim}{r}
\newcommand{\U}{\mathbf{U}}
\newcommand{\I}{\mathbf{I}}
\renewcommand{\v}{\mathbf{v}}
\newcommand{\LP}[1]{\mathrm{LP}(#1)}
\newcommand{\LC}[1]{\mathrm{LC}(#1)}
\newcommand{\LT}[1]{\mathrm{LT}(#1)}

\newcommand{\m}{\mathbf{m}}
\newcommand{\E}{\mathbf{E}}
\newcommand{\V}{\mathbf{V}}

\newcommand{\B}{\mathbf{B}}
\newcommand{\Redhat}{\hat{\Red}}

\DeclareMathOperator{\len}{len}

\begin{document}

\begin{verbatim}\end{verbatim}\vspace{2.5cm}

\begin{frontmatter}

\title{Decoding Interleaved Gabidulin Codes using Alekhnovich's Algorithm}

\author{Sven Puchinger$^\mathrm{a}$, Sven M\"uelich$^\mathrm{a}$, David M\"odinger$^\mathrm{b}$,}
\author{Johan Rosenkilde n\'e Nielsen$^\mathrm{c}$ and Martin Bossert$^\mathrm{a}$\thanksref{emails}}

\address{$^\mathrm{a}$Institute of Communications Engineering, Ulm University, Ulm, Germany}
\address{$^\mathrm{b}$Institute of Distributed Systems, Ulm University, Ulm, Germany} 
\address{$^\mathrm{c}$Department of Applied Mathematics \& Computer Science, Technical University of Denmark,
Lyngby, Denmark}
\thanks[emails]{Email:
   \href{mailto:sven.puchinger@uni-ulm.de} {\texttt{\normalshape sven.puchinger@uni-ulm.de}},
   \href{mailto:sven.mueelich@uni-ulm.de} {\texttt{\normalshape sven.mueelich@uni-ulm.de}},\\
   \href{mailto:david.moedinger@uni-ulm.de} {\texttt{\normalshape david.moedinger@uni-ulm.de}},
   \href{mailto:jsrn@jsrn.dk} {\texttt{\normalshape jsrn@jsrn.dk}},
   \href{mailto:martin.bossert@uni-ulm.de} {\texttt{\normalshape martin.bossert@uni-ulm.de}},
}

\begin{abstract}
We prove that Alekhnovich's algorithm can be used for row reduction of skew polynomial matrices. This yields an $O(\ell^3 n^{(\omega+1)/2} \log(n))$ decoding algorithm for $\ell$-Interleaved Gabidulin codes of length $n$, where $\omega$ is the matrix multiplication exponent, improving in the exponent of $n$ compared to previous results.
\end{abstract}


\end{frontmatter}

\section{Introduction}
\label{sec:intro}
It is shown in \cite{puchinger2015row} that \emph{Interleaved Gabidulin codes} of \emph{length} $n \in \NN$ and \emph{interleaving degree} $\ell \in \NN$ can be error- and erasure-decoded by transforming the following \emph{skew polynomial} \cite{ore1933theory} matrix into \emph{weak Popov form} (cf.~Section~\ref{sec:preliminaries})\footnote{Afterwards, the corresponding information words are obtained by $\ell$ many divisions of skew polynomials of degree $\BigOtext{n}$, which can be done in $\BigOtext{\ell n^{(\omega+1)/2} \log(n)}$ time \cite{puchinger2016fast}.\label{fn:test}}:
\vspace{-0.3cm}
\begin{align}\renewcommand{\arraystretch}{0.3}
\B =
\begin{bmatrix}
x^{\gamma_0} & s_1 x^{\gamma_1} & s_2 x^{\gamma_2} & \dots & s_\ell x^{\gamma_\ell}\\
0 & g_1 x^{\gamma_1} & 0 & \dots  & 0 \\
0 & 0 & g_2 x^{\gamma_2} & \dots  & 0 \\
\vdots & \vdots & \vdots & \ddots&\vdots \\
0 & 0 & 0 & \dots  & g_\ell x^{\gamma_\ell} \\
\end{bmatrix}, \label{eq:B}
\end{align}
\vspace{-0.6cm}

\noindent
where the skew polynomials $s_1,\dots,s_\ell,g_1,\dots,g_\ell$ and the non-negative integers $\gamma_0,\dots,\gamma_\ell$ arise from the decoding problem and are known at the receiver.
Due to lack of space, we cannot give a description of Interleaved Gabidulin codes, the mentioned procedure and the resulting decoding radius here and therefore refer to \cite[Section~3.1.3]{puchinger2015row}.
By adapting row reduction\footnote{By row reduction we mean to transform a matrix into weak Popov form by row operations.} algorithms known for polynomial rings $\F[x]$ to skew polynomials, a decoding complexity of $O(\ell n^2)$ can be achieved \cite{puchinger2015row}.
In this paper, we adapt Alekhnovich's algorithm \cite{alekhnovich2005linear} for row reduction of $\F[x]$ matrices to the skew polynomial case.

\section{Preliminaries}
\label{sec:preliminaries}

Let $\F$ be a finite field and $\sigma$ an $\F$-automorphism.
A \emph{skew polynomial ring} $\R$ \cite{ore1933theory} contains polynomials of the form $a = \sum_{i=0}^{\deg a} a_i x^i$, where $a_i \in \F$ and $a_{\deg a} \neq 0$ ($\deg a$ is the \emph{degree} of $a$), which are multiplied according to the rule $x \cdot a = \sigma(a) \cdot x$, extended recursively to arbitrary degrees.
This ring is non-commutative in general.
All polynomials in this paper are skew polynomials.

It was shown in \cite{wachter2013decoding} for linearized polynomials and generalized in \cite{puchinger2016fast} to arbitrary skew polynomials that two such polynomials of degrees $\leq s$ can be multiplied with complexity $\OMul{s} \in \BigOtext{s^{(\omega+1)/2}}$ in operations over $\F$, where $\omega$ is the matrix multiplication exponent.

A polynomial $a$ has \emph{length} $\len a$ if $a_i = 0$ for all $i=0,\dots,\deg a - \len a$ and $a_{\deg a - \len a+1} \neq 0$.
We can write $a = \tilde{a} x^{\deg a - \len a+1}$, where $\deg \tilde{a} \leq \len a$, and multiply $a,b \in \R$ by
$a \cdot b = [\tilde{a} \cdot \sigma^{\deg a - \len a+1}(\tilde{b})] x^{\deg a + \deg a - \len a - \len b + 1}.$
Computing $\sigma^i(\alpha)$ with $\alpha \in \F$, $i \in \NN$ is in $O(1)$ (cf.~\cite{puchinger2016fast}).
Hence, $a$ and $b$ of length $s$ can be multiplied in $\OMul{s}$ time, although possibly $\deg a, \deg b \gg s$.

Vectors $\v$ and matrices $\M$ are denoted by bold and small/capital letters.
Indices start at $1$, e.g. $\v = (v_1,\dots,v_\mdim)$ for $\mdim \in \NN$.
$\E_{i,j}$ is the matrix containing only one non-zero entry $=1$ at position $(i,j)$ and $\I$ is the identity matrix.
We denote the $i$th row of a matrix $\M$ by $\m_i$.
The degree of a vector $\v \in \R^\mdim$ is the maximum of the degrees of its components $\deg \v = \max_i\{\deg v_i\}$ and the degree of a matrix $\M$ is the sum of its rows' degrees $\deg \M = \sum_i \deg \m_i$.

The \emph{leading position} (LP) of $\v$ is the rightmost position of maximal degree $\LP{\v} = \max \{i : \deg v_i = \deg \v\}$.
The \emph{leading coefficient} (LC) of a polynomial $a$ is $\LT{a} = a_{\deg a} x^{\deg a}$ and the \emph{leading term} (LT) of a vector $\v$ is $\LT{\v} = v_{\LP{\v}}$.
A matrix $\M \in \R^{\mdim \times \mdim}$ is in \emph{weak Popov form} (wPf) if the leading positions of its rows are pairwise distinct.
E.g., the following matrix is in wPf since $\LP{\m_1} = 2$ and $\LP{\m_2} = 1$
\vspace{-0.3cm}
\begin{align*}\renewcommand{\arraystretch}{0.3}
\M = \begin{bmatrix}
x^2+x & x^2+1 \\
x^4 & x^3+x^2+x+1
\end{bmatrix}.
\end{align*}
\vspace{-0.7cm}

Similar to \cite{alekhnovich2005linear}, we define an \emph{accuracy approximation to depth} $t \in \NN_0$ of skew polynomials as $a|_t = \sum_{i=\deg a-t+1}^{\deg a} a_i x^i$.
For vectors, it is defined as $\v|_t = (v_1|_{\min\{0, t-(\deg \v - \deg v_1)\}},\dots,v_\mdim|_{\min\{0, t-(\deg \v - \deg v_\mdim)\}})$ and for matrices row-wise.
E.g., with $\M$ as above,
\vspace{-0.3cm}
\begin{align*}\renewcommand{\arraystretch}{0.3}
\M|_2 = \begin{bmatrix}
x^2+x & x^2 \\
x^4 & x^3
\end{bmatrix}
\text{ and }
\M|_1 = \begin{bmatrix}
x^2 & x^2 \\
x^4 & 0
\end{bmatrix}.
\end{align*}
\vspace{-0.7cm}

We can extend the definition of the length of a polynomial to vectors $\v$ as $\len \v = \max_i\{\deg \v - \deg v_i + \len v_i\}$ and to matrices as $\len \M = \max_i\{\len \m_i\}$.
With this notation, we have $\len(a|_t) \leq t$, $\len(\v|_t) \leq t$ and $\len(\M|_t) \leq t$.

\section{Alekhnovich's Algorithm over Skew Polynomials}

Alekhnovich's algorithm \cite{alekhnovich2005linear} was proposed for transforming matrices over ordinary polynomials $\F[x]$ into wPf.
Here, we show that, with a few modifications, it also works with skew polynomials.
As in the original paper, we prove the correctness of Algorithm~\ref{alg:recursion_dc} (main algorithm) using the auxiliary Algorithm~\ref{alg:base}.

\begin{alg}$\Red(\M)$ \\
\label{alg:base}
Input: Module basis $\M \in \R^{\mdim \times \mdim}$ with $\deg \M = n$ \\
Output: $\U \in \R^{\mdim \times \mdim}$: $\U \cdot \M$ is in wPf or $\deg(\U \cdot \M ) \leq \deg \M-1$\\
1. $\U \gets \I$\\
2. While $\deg \M = n$ and $\M$ is not in wPf \\
3. ~~~Find $i,j$ such that $\LP{\m_i} = \LP{\m_j}$ and $\deg \m_i \geq \deg \m_j$ \label{line:rb:st_start}\\
4. ~~~$\delta \gets \deg \m_i - \deg \m_j$ and $\alpha \gets \LC{\LT{\m_i}}/\theta^\delta (\LC{\LT{\m_j}})$\\
5.~~~~$\U \gets (\I -\alpha x^\delta \E_{i,j}) \cdot \U$ and
$\M \gets (\I -\alpha x^\delta \E_{i,j}) \cdot \M$ \label{line:rb:TB} \\
6. Return $\U$\\
\end{alg}

\begin{theorem}\label{thm:base_alg}
Algorithm~\ref{alg:base} is correct and if $\len(\M) \leq 1$, it is in $\BigO{\mdim^3}$.
\end{theorem}

\begin{proof}
Inside the while loop, the algorithm performs a so-called \emph{simple transformation} (ST).
It is shown in \cite{puchinger2015row} that such an ST on an $\R$-matrix $\M$ preserves both its rank and row space (this does not trivially follow from the $\F[x]$ case due to non-commutativity) and reduces either $\LP{\m_i}$ or $\deg \m_i$.
At some point, $\M$ is in wPf, or $\deg \m_i$ and likewise $\deg \M$ is reduced by one.
The matrix $\U$ keeps track of the STs, i.e. multiplying $\M$ by $(\I -\alpha x^\delta \E_{i,j})$ from the left is the same as applying an ST on $\M$.
At termination, $\M = \U \cdot \M'$, where $\M'$ is the input matrix of the algorithm.
Since $\sum_i \LP{\m_i}$ can be decreased at most $\mdim^2$ times without changing $\deg \M$, the algorithm performs at most $\mdim^2$ STs.
Multiplying $(\I -\alpha x^\delta \E_{i,j})$ by a matrix $\V$ consists of scaling a row with $\alpha x^\delta$ and adding it to another (target) row.
Due to the accuracy approximation, all monomials of the non-zero polynomials in the scaled and the target row have the same power, implying a cost of $\mdim$ for each ST.
The claim follows.
\end{proof}

We can decrease a matrix' degree by at least $t$ or transform it into wPf by $t$ recursive calls of Algorithm~\ref{alg:base}. We can write this as $\Red(\M,t) = \U \cdot \Red(\U \cdot \M)$, where $\U = \Red(\M,t-1)$ for $t>1$ and $\U = \I$ if $t=1$.
As in \cite{alekhnovich2005linear}, we speed this method up by two modifications.
The first one is a divide-\&-conquer (D\&C) trick, where instead of reducing the degree of a ``$(t-1)$-reduced'' matrix $\U \cdot \M$ by $1$ as above, we reduce a ``$t'$-reduced'' matrix by another $t-t'$ for an arbitrary $t'$. For $t' \approx t/2$, the recursion tree has a balanced workload.
\begin{lemma}\label{lem:recursion}
Let $t'<t$ and $\U = \Red(\M,t')$. Then, 
\vspace{-0.3cm}
\begin{align*}
\Red(\M,t) = \Red\big[\U \cdot \M, t-(\deg \M - \deg(\U \cdot \M))\big] \cdot \U.
\end{align*}
\end{lemma}

\begin{proof}
$\U$ reduces reduces $\deg \M$ by at least $t'$ or transforms $\M$ into wPf.
Multiplication by $\Red[\U \cdot \M, t-(\deg \M - \deg(\U \cdot \M))]$ further reduces the degree of this matrix by $t-(\deg \M - \deg(\U \cdot \M)) \geq t-t'$ (or $\U \cdot \M$ in wPf).
\end{proof}

The second lemma allows to compute only on the top coefficients of the input matrix inside the divide-\&-conquer tree, reducing the overall complexity.

\begin{lemma}\label{lem:approx}
$\Red(\M,t) = \Red(\M|_t,t)$
\end{lemma}

\begin{proof}
Arguments completely analogous to the $\F[x]$ case of \cite[Lemma~2.7]{alekhnovich2005linear} hold.
\end{proof}

\begin{lemma}\label{lem:Ulength}
$\Red(\M,t)$ contains polynomials of length $\leq t$.
\end{lemma}

\begin{proof}
The proof works as in the $\F[x]$ case, cf.\ \cite[Lemma~2.8]{alekhnovich2005linear}, by taking care of the fact that $\alpha x^a \cdot \beta x^b = \alpha \sigma^c(\beta) x^{a+b}$ for all $\alpha,\beta \in \F$, $a,b \in \NN_0$. 
\end{proof}

\begin{alg}$\Redhat(\M,t)$ \\
Input: Module basis $\M \in \R^{\mdim \times \mdim}$ with $\deg \M = n$ \\
Output: $\U \in \R^{\mdim \times \mdim}$: $\U \cdot \M$ is in wPf or $\deg(\U \cdot \M ) \leq \deg \M-t$ \\
1. If $t=1$, then Return \(\Red(\M|_1)\)\\
2. $\U_1 \gets \Redhat(\M|_t,\lfloor t/2\rfloor)$ and $\M_1 \gets \U_1 \cdot \M|_t$ \\
3. Return $\Redhat(\M_1, t - (\deg \M|_t - \deg \M_1)) \cdot \U_1$ \label{line:matrix:mul:2}
\label{alg:recursion_dc}
\end{alg}

\begin{theorem}\label{thm:main}
Algorithm~\ref{alg:recursion_dc} is correct and has complexity $\BigOtext{\mdim^3 \OMul{t}}$.
\end{theorem}

\begin{proof}
Correctness follows from $\Red(\M,t) = \Redhat(\M,t)$ by induction (for $t=1$, see Theorem~\ref{thm:base_alg}). Let $\hat{\U} = \Redhat(\M|_t,\lfloor \tfrac{t}{2}\rfloor)$ and $\U = \Red(\M|_t,\lfloor \tfrac{t}{2}\rfloor)$. Then,
\vspace{-0.3cm}
\begin{align*}
\Redhat(\M,t) &= \Redhat(\hat{\U} \cdot \M|_t, t - (\deg \M|_t - \deg (\hat{\U} \cdot \M|_t))) \cdot \hat{\U} \\
\overset{\mathrm{(i)}}{=} &\, \Red(\U \cdot \M|_t, t - (\deg \M|_t - \deg (\U \cdot \M|_t))) \cdot \U \overset{\mathrm{(ii)}}{=} \Red(\M|_t,t) \overset{\mathrm{(iii)}}{=} \Red(\M,t),
\end{align*}
\vspace{-0.8cm}

\noindent
where (i) follows from the induction hypothesis, (ii) by Lemma~\ref{lem:recursion}, and (iii) by Lemma~\ref{lem:approx}.
Algorithm~\ref{alg:recursion_dc} calls itself twice on inputs of sizes $\approx \tfrac{t}{2}$.
The only other costly operations are the matrix multiplications in Lines~2 and 3 of matrices containing only polynomials of length $\leq t$ (cf. Lemma~\ref{lem:Ulength}).
This costs\footnote{In D\&C matrix multiplication algorithms, the length of polynomials in intermediate computations might be much larger than $t$. Thus, we have to compute it naively in cubic time.} $\mdim^2$ times $\mdim$ multiplications $\OMul{t}$ and $\mdim^2$ times $\mdim$ additions $O(t)$ of polynomials of length $\leq t$, having complexity $\BigOtext{\mdim^3 \OMul{t}}$.
The recursive complexity relation reads
$f(t) = 2 \cdot f(\tfrac{t}{2}) + O(\mdim^3 \OMul{t}).$
By the master theorem, we get $f(t) \in \BigOtext{t f(1) + \mdim^3 \OMul{t}}$.
The base case operation $\Red(\M|_1)$ with cost $f(1)$ is called at most $t$ times since it decreases $\deg \M$ by $1$ each time. Since $\len(\M|_1)\leq 1$, $f(1) \in \BigO{\mdim^3}$ by Theorem~\ref{thm:base_alg}.
Hence, $f(t) \in \BigOtext{\mdim^3 \OMul{t}}$.
\end{proof}

\section{Implications and Conclusion}

The \emph{orthogonality defect} \cite{puchinger2015row} of a square, full-rank, skew polynomial matrix~$\M$ is $\Delta(\M) = \deg \M - \deg \det \M$, where $\deg \det$ is the ``determinant degree'' function, see~\cite{puchinger2015row}.
A matrix $\M$ in wPf has $\Delta(\M)=0$ and $\deg \det \M$ is invariant under row operations.
Thus, if $\V$ is in wPf and obtained from $\M$ by simple transformations, then
$\deg \V = \Delta(\V) + \deg \det \V = \deg \M - \Delta(\M)$.
With $\Delta(\M) \geq 0$, this implies that $\Redhat(\M, \Delta(\M)) \cdot \M$ is always in wPf.
It was shown in \cite{puchinger2015row} that $\B$ from Equation \eqref{eq:B} has orthogonality defect $\Delta(\B) \in O(n)$, which implies the following theorem.
\begin{theorem}[Main Statement]
$\Redhat(\B,\Delta(\B)) \cdot \B$ is in wPf. This implies that we can decode Interleaved Gabidulin codes in\footnote{The $\log(n)$ factor is due to the divisions in the decoding algorithm, following the row reduction step (see Footnote~\ref{fn:test}) and can be omitted if $\log(n) \in o(\ell^2)$.} $\BigOtext{\ell^3 n^{(\omega+1)/2} \log(n)}$.
\end{theorem}
Table~\ref{tab:comparison} compares the complexities of known decoding algorithms for Interleaved Gabidulin codes.
Which algorithm is asymptotically fastest depends on the relative size of $\ell$ and $n$.
Usually, one considers $n \gg \ell$, in which case the algorithms in this paper and in \cite{sidorenko2014fast} provide---to the best of our knowledge---the fastest known algorithms for decoding Interleaved Gabidulin codes.

\begin{table}[h!]
\centering
\renewcommand{\arraystretch}{1.2}
\footnotesize
\begin{tabular}{l|l}
Algorithm & Complexity  \\[0.1cm]
\hline \hline
Skew Berlekamp--Massey \cite{sidorenko2011skew} & $O(\ell n^2)$ \\[0.1cm]
\hline
Skew Berlekamp--Massey (D\&C) \cite{sidorenko2014fast}&  $O(\ell^K n^{\frac{\omega+1}{2}} \log(n))$, possibly\footnotemark $K=3$\\[0.1cm]
\hline 
Skew Demand--Driven$^\ast$ \cite{puchinger2015row} & $O(\ell n^2)$ \\[0.1cm] 
\hline 
Skew Alekhnovich$^\ast$ (Theorem~\ref{thm:main}) & $O(\ell^3 n^{\frac{\omega+1}{2}} \log(n)) \subseteq^{\dagger} O(\ell^3 n^{1.69} \log(n))$ \\
\hline 
\end{tabular}
\caption{Comparison of decoding algorithms for Interleaved Gabidulin codes. Algorithms marked with $^\ast$ are based on the row reduction problem of \cite{puchinger2015row}. $^{\dagger}$Example $\omega \approx 2.37$.}
\label{tab:comparison}
\end{table}
\footnotetext{In \cite{sidorenko2014fast}, the complexity is given as $O(n^{\frac{\omega+1}{2}} \log(n))$ and $\ell$ is considered to be constant. By a rough estimate, the complexity becomes $O(\ell^{O(1)} n^{\frac{\omega+1}{2}} \log(n))$ when including $\ell$. We believe the exponent of $\ell$ is really 3 (or possibly $\omega$) but this should be further analyzed.}

In the case of Gabidulin codes ($\ell=1$), we obtain an alternative to the \emph{Linearized Extended Euclidean} algorithm from \cite{wachter2013decoding} of the same complexity.
The algorithms are equivalent up to the implementation of a simple transformation.

\end{document}